\documentclass[12pt]{article}
\usepackage[mathscr]{eucal}
\usepackage{amsthm,amsmath,amssymb,amscd} 
\usepackage{graphicx}
\usepackage{latexsym}
\usepackage[dvips]{color}
\newtheorem{theo}{Theorem}[section]

\newtheorem{lemma}{Lemma}[section]
\newtheorem{cor}{Corollary}[section]\newtheorem{prop}{Proposition}[section]
\theoremstyle{definition}
\newtheorem{remark}{Remark}[section]

\title{The Kepler Problem with Anisotropic Perturbations}

\author{}\def \beq{ \begin{equation} }
\def \eeq{\end{equation}}
\def \r#1{$^{#1}$}

\def \d{{\rm d}}

\def \pd{\partial}

\def\.#1{\dot #1}

\def \th{\theta}
\def \Th{\Theta}
\def \ep{\epsilon}

\def \r#1{$^{#1}$}
\def \( {\big( }
\def \) {\big) }
\def \intR {\int_{-\infty}^{+\infty}}
\def \grad {\nabla}
\def \ov{\over}

\def \qq {\qquad}
\def \bar{\overline}
\def \1{{\bf I}}
\def \2{{\bf II}}
\def\.#1{\dot #1}
\date{}
\begin{document}
\DeclareGraphicsRule{.eps.gz}{eps}{.eps.bb}{#1}
\DeclareGraphicsRule{.ps.gz}{eps}{.ps.bb}{#1}
\maketitle
\author{\begin{center}Florin Diacu\\
\end{center}
\begin{center}
{\footnotesize 
Department of Mathematics and Statistics\\
University of Victoria\\
P.O.~Box 3045 STN CSC\\
Victoria, BC, Canada, V8W 3P4\\
diacu@math.uvic.ca\\
}\end{center}

\begin{center}
Ernesto P\'erez-Chavela\\
\end{center}
\begin{center}
{\footnotesize
Departamento de Matem\'aticas\\
Universidad Aut\'onoma Metropolitana-Iztapalapa\\
Apdo.\ 55534, M\'exico, D.F., M\'exico\\
epc@xanum.uam.mx\\
}\end{center}

\begin{center}
Manuele Santoprete\\
\end{center}
\begin{center}
{\footnotesize Department of Mathematics\\
University of California, Irvine\\
Irvine CA, 92697 USA\\
msantopr@math.uci.edu\\
}\end{center}
}

\begin{abstract}
We study a 2-body problem given by the sum of the Newtonian
potential and an anisotropic perturbation that is a homogeneous
function of degree $-\beta$, $\beta\ge 2$. For $\beta>2$, the sets
of initial conditions leading to collisions/ejections and the one
leading to escapes/captures have positive measure. For $\beta>2$
and $\beta\ne 3$, the flow on the zero-energy manifold is chaotic.
For $\beta=2$, a case we prove integrable, the infinity manifold
of the zero-energy level is a disconnected set, which has
heteroclinic connections with the collision manifold.
\end{abstract}

\vskip 0.5truecm
\hspace{25pt}  PACS(2003): 05.45.-a, 45.50.Jf, 45.50.Pk

\vfill\eject



\section*{\large\bf I. Introduction}

In the past three centuries, celestial mechanics has stimulated the
development of many branches of mathematics (see e.g. Refs. 11,13). This trend continues, and even its most
basic question, known as the {\it 2-body problem} or the {\it Kepler
problem}, still attracts the vivid interest of mathematicians and
physicists, both in its classical form (see Ref. 1) and in newer
versions.

Among the latter are the problems raised by {\it quasihomogeneous} potentials,
given by the sum of homogeneous functions, and problems set in {\it aniso\-tropic} spaces, for which the interaction law is different in each direction of the space. For quasihomogeneous problems the terminology and the first qualitative results were introduced in the mid 1990s;\r{9,12,21} this potential unifies several dynamical laws, including those of Newton, Coulomb,
Manev, Schwarzschild, Lennard-Jones, Birkhoff and others. The anisotropic
case is more related to physics, and was initiated by Martin Gutzwiller in
the 1970s\r{17,18}  for the quantization
of classical ergodic systems. Among Gutzwiller's goals was also that of
finding connections between classical and quantum mechanics. A combination
of the quasihomogeneous and anisotropic aspects shows up in the anisotropic Manev problem, whose dynamics contains classical, quantum and relativistic features (see Refs. 6,10,14,15).

In the present paper we consider a new version of the Kepler problem, which combines two of the above characteristics: isotropy and anisotropy. The potential (see formula (\ref{potential})) is the sum of the classical
Keplerian potential and an anisotropic perturbation, the latter being an homogeneous function of degree $-\beta$, $\beta\ge 2$. This is the first analysis of a {\it quasihomogeneous} potential that mixes isotropic and
anisotropic components. For previously studied problems, all terms have
been either isotropic or anisotropic. As we will see, this case has some
surprising dynamical properties, often very different from the ones that characterize potentials whose terms are not mixed.

Such mixed potentials can be used to understand the dynamics of
satellites around oblate planets or the motion of stars around black holes. Here, however, we are not interested in applications. Our endeavours are restricted to mathematical results.

In Sec. II we introduce the basic notations, the equations of motion,
and put into the evidence the symmetries of the problem.
In Sec. III we begin the study of the case $\beta>2$, define the collision manifold, which is an essential qualitative tool, and perform a geometric
study of the flow in the neighbourhood of collisions. We classify all collision-ejection orbits and prove that the set of initial conditions
leading to them has positive measure. We achieve this while studying the flow on and near the collision manifold in terms of McGehee-type coordinates.

In Sec. IV we investigate the existence of heteroclinic orbits on the collision manifold for potentials with $\mu>1$ and $\beta=2+{2\over 1+2k}$ or $\beta=2+{1\over 1+k}$. The main result of this section is that for an
open and dense set of $\mu$-values, saddle-saddle connections do not exist
on the collision manifold. Section V deals with capture and escape orbits
in the zero-energy case. We show that the infinity manifold has
two circles of normally hyperbolic equilibria, one attractive and
one repelling. This proves the existence of infinitely many
capture and escape orbits.

In Sec. VI we consider the case $\beta=2$, which we show to be
integrable. Apparently this is quite surprising since the anisotropic
Manev problem, which resembles this case except for the anisotropy of
the Newtonian term, is nonintegrable (see Refs. 3,15). But as we will show, the surprise element vanishes once we look at the problem in the larger context of the Hamilton-Jacobi theory. In Sec. VII we study the flow on and near
the collision manifold and see that its qualitative behaviour is
similar to that of the anisotropic Manev problem.

In Sec. VIII we prove the existence of heteroclinic orbits connecting
the collision and infinity manifolds in the zero-energy case for $\beta=2$.
The infinity manifold is formed by two disjoint circles of equilibria,
which seems very unusual. In fact this is the first case we have ever encountered in which the infinity (or collision) manifold is a disconnected set. A complete understanding of the topological structure of the phase space
when all parameters are varied will probably shed more light in this matter.
But a study in this direction, though interesting and worthy of further
investigations, is beyond our present scope.

In Sec. IX we consider a perturbative approach of the problem. The perturbation function of the Hamiltonian is a homogeneous function of degree $-\beta$ with $\beta>3/2$. We end the paper with Section X, in which we apply the Melnikov method to show that for every $\beta\ne 2, 3$, the flow on
the zero-energy manifold is chaotic. This also proves the nonintegrability
of the problem for $\beta\ne 2, 3$.

\section*{\large\bf II. Equations of Motion and Symmetries}

Consider the Hamiltonian
\begin{equation}\label{hamf}
   H_{\beta}({\bf q}, {\bf p})\,=\,\frac{1}{2}\,{\bf p}^2 + U_{\beta}({\bf q}),
\end{equation}
where ${\bf q}=(x,y)$ and ${\bf p}=(p_x,p_y)$.
The equations of motion are
\beq
\begin{split}
\dot{\bf q}=&~{\bf p}\\
\dot{\bf p}=&-\grad U_{\beta}({\bf q}),
\end{split}
\label{eqmotion}
\eeq
where $U_{\beta}$ is a potential of the form
\beq
U_{\beta}(x,y) = -\frac{1}{\sqrt{x^2 + y^2}} - \frac{b}{(\mu x^2 + y^2)^{\beta /2}},
\label{potential}
\eeq
with the constants $beta\ge 2$, $\mu\ge 1$ and $b>0$.
The symmetries of (\ref{eqmotion}) are given by the following analytic diffeomorphisms in the extended phase space:
\beq \begin {array}{l} \
\hspace{2pt}      Id \hspace{.6pt} :~(x,y,p_x,p_y,t)\longrightarrow(x,y,p_x,p_y,t),\\\
    S_0~:~(x,y,p_x,p_y,t)\longrightarrow(x,y,-p_x,-p_y,-t),
\\\
S_1~:~(x,y,p_x,p_y,t)\longrightarrow(x,-y,-p_x,p_y,-t),\\\
S_2~:~(x,y,p_x,p_y,t)\longrightarrow(-x,y,p_x,-p_y,-t),\\\
S_3~:~(x,y,p_x,p_y,t)\longrightarrow(-x,-y,-p_x,-p_y,t),\\\
S_4~:~(x,y,p_x,p_y,t)\longrightarrow(-x,y,-p_x,p_y,t),\\\
S_5~:~(x,y,p_x,p_y,t)\longrightarrow(x,-y,p_x,-p_y,t),\\\
S_6~:~(x,y,p_x,p_y,t)\longrightarrow(-x,-y,p_x,p_y,-t),
\label{symmetries}
\end {array}
\eeq where $Id$ is the identity. These diffeomorphisms form a group that is isomorphic to ${\mathbb Z}_2\times {\mathbb Z}_2\times \mathbb{Z}_2$ (see Refs. 14,22 for more details).
Invariance under these symmetries implies that if  $\gamma(t)$  is a
solution of (\ref{eqmotion}), then also $S_i(\gamma(t))$ is solution
for $i \in \{0,1,2,3,4,5,6\}$.

\section*{\large\bf III. The Collision Manifold for $\beta>2$}

We will further express the equations of motion in McGehee-type
coordinates,\r{8,20} which are suitable for understanding the motion
near collision.
The transformations are given step by step as follows. Take
\begin{eqnarray}\label{mcgehee}
r &=& x^2 + y^2,  \nonumber \\
\theta &=& \arctan{\frac{y}{x}},  \\
\tilde{v} &=& r\dot{r} = (xp_x + yp_y),  \nonumber \\
\tilde{u} &=& r^{2}\dot{\theta} = (xp_y - yp_x),\nonumber
\end{eqnarray}
rescale $\tilde{v}$ and $\tilde{u}$ by
$$v=r^{\frac{\beta - 2}{2}}\tilde{v}, \quad u=r^{\frac{\beta -
2}{2}}\tilde{u},$$
and then rescale the time variable using the transformation
$$ \frac{dt}{d\tau} = r^{\beta/2 + 1}.$$
The equations of motion take the form
\begin{eqnarray}\label{eqmot}
r^{\prime } &=& rv,  \nonumber \\
v^{\prime } &=&\frac{\beta -2}{2}v^{2} + r^{\beta-1} +
2hr^{\beta}- \frac{b(\beta -2)}{\Delta^{\beta/2}}, \nonumber \\
\theta^{\prime } &=& u,  \\
{u}^{\prime } &=& \frac{\beta -2}{2}uv + \frac{b\beta(\mu
-1)\sin{2\theta}}{2\Delta^{\frac{\beta +2}{2}}}, \nonumber
\end{eqnarray}
where $\Delta= \mu \cos^2{\theta} + \sin^2{\theta}$ and the
prime denotes differentiation with respect to the new independent
time variable $\tau $. For simplicity, we keep the same notation
for the new dependent variables.

In these new coordinates, the energy integral $H_{\beta}=h$ (see equation (\ref{hamf})), takes the form
\begin{equation}\label{eqenerg}
u^2 + v^2 - 2r^{\beta -1} - \frac{2b}{\Delta^{\beta/2}} =
2hr^{\beta}.
\end{equation}

We define the collision manifold (see Figure 1) as
\begin{equation}\label{totcol}
C = \{ (r, v, \theta, u) | \quad  r =0, \quad u^2 + v^2 =
\frac{2b}{\Delta^{\beta/2}} \}.
\end{equation}
Notice that $C$ is homeomorphic to a torus.
The flow on this manifold is given by the system
\begin{eqnarray}\label{eqcol}
v^{\prime } &=& \frac{\beta -2}{2}(-u^{2}), \nonumber \\
\theta^{\prime } &=& u,  \\
{u}^{\prime } &=& \frac{\beta -2}{2}uv + \frac{b\beta(\mu
-1)\sin{2\theta}}{2\Delta^{\frac{\beta +2}{2}}}. \nonumber
\end{eqnarray}

\begin{figure}[t]
\begin{center}
\resizebox{!}{7cm}{\includegraphics{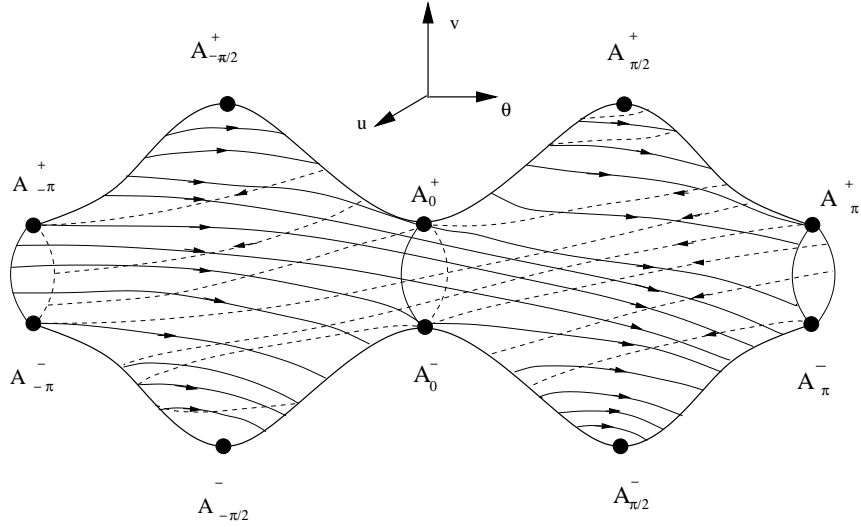}}
\end{center}
\caption{The collision manifold $C$}%
\label{collision}
\end{figure}

We can now prove the following results.

\begin{prop} All the equilibrium points of system (\ref{eqmot})
lie on the collision manifold $C$ and are given by
$$r=0, \quad v=\pm \sqrt{\frac{2b}{\Delta^{\frac{\beta}{2}}}},
\quad \theta= 0,\pi/2,\pi, 3\pi/2, \quad u=0.$$\end{prop}

\begin{proof}
It is obvious that the above points are equilibria for the flow defined
by (\ref{eqmot}). To check that there are no equilibria outside the
collision manifold, from the first equation in (\ref{eqmot}) with $r\neq 0$
we see that $v=0$ and from the third equation in (\ref{eqmot}), that $u=0$. Substituting these values in (\ref{eqenerg}), we obtain
\begin{equation}\label{first}
r^{\beta -1}
+ 2hr^{\beta} = \frac{2b}{\Delta^{\beta/2}}.
\end{equation}
Solving the equation $v^{\prime } =0 $ in (\ref{eqmot}), we are led to
\begin{equation}\label{second}
r^{\beta-1} + 2hr^{\beta}= \frac{b(\beta -2)}{\Delta^{\beta/2}}.
\end{equation}
But (\ref{first}) and (\ref{second}) have no common solutions,
therefore there are no equilibria outside the collision manifold.
\end{proof}

\begin{prop} The flow on the collision manifold $C$ is {\it gradient-like} (i.e.\ increasing) with respect to the $(-v)$-coordinate.
\end{prop}
\begin{proof}
Since $\beta >2$, we see from the first equation in (\ref{eqcol})
that $v^{\prime}<0$ except at equilibria, therefore the flow on
$C$ increases with respect to $-v$, so is gradient-like relative
to it.
\end{proof}

To match the sign of $v$ and the value of $\theta$, we denote
the equilibria on $C$ by $A_{0}^{\pm}, A_{\pi/2 }^{\pm}, A_{\pi }^{\pm }$
and $ A_{3\pi/2}^{\pm}$, respectively. Observe that $\Delta
(0) = \Delta (\pi) = \mu$ and $\Delta (\pi/2) = \Delta (3\pi/2) =
1$. With this notation, we can describe the following properties of the flow.

\begin{theo} On the collision manifold $C$, the equilibria $
A_{0}^{\pm }$ and $A_{\pi }^{\pm }$ are saddles, $A_{\pi /2}^{+}$
and $A_{3\pi /2}^{+}$ are sources and $A_{\pi /2}^{-}$ and
$A_{3\pi /2}^{-}$ are sinks. Outside $C$, the equilibria
$A_{0}^{+}$, $A_{\pi }^{+}$, $A_{\pi /2}^{+}$ and $A_{3\pi
/2}^{+}$ have a local one-dimensional unstable analytic manifold,
whereas $A_{0}^{-}$, $A_{\pi }^{-}$, $A_{\pi /2}^{-}$ and $A_{3\pi
/2}^{-}$ have a local one-dimensional stable analytic manifold.
All these equilibrium points are hyperbolic.
\end{theo}

\begin{proof}
Consider the function
$$F(r,v,\theta ,u) = u^2 + v^2 - r^{\beta -1} - \frac{2b}{\Delta^{\beta/2}}
-2hr^{\beta}= 0.$$
According to equation (\ref{eqenerg}), the surface of constant energy $M_{h}$
defined by the equation  $$F(r,\theta ,v,u)=0$$ is a 3-dimensional  manifold. At every point $B$ of $M_h$, the tangent space is given by
$$T_{B}F=\{(r,v,\theta ,u)|\nabla F(B)\cdot (r,v,\theta,u)=0\}.$$
At any equilibrium point $A$, the tangent space is defined by
$$ T_AF=\{(r,\theta,v,u)| \; v=0\}.$$

A straightforward computation shows that at the equilibria $A_{0}^{\pm }$ and
$A_{\pi }^{\pm }$ the linearized system corresponding to
(\ref{eqmot}) has the matrix
\[
\left[
\begin{array}{cccc}
\pm \sqrt{\frac{2b}{\mu ^{\beta/2}}} & 0 & 0 & 0 \\
0 & (\beta - 2) \pm (\sqrt{ \frac{2b}{\mu ^{\beta /2}} }) & 0 & 0 \\
0 & 0 & 0 & 1 \\
0 & 0 & \frac{b \beta (\mu -1)}{\mu ^{\frac{\beta + 2}{2}}} & \pm
\frac{(\beta - 2)}{2}( \sqrt{\frac{2b}{\mu ^{\beta/2}}})
\end{array}
\right] ,
\]
therefore the linear part of the vector field (\ref{eqmot})
restricted to $ T_{A_{0,\pi }}^{+,-}$ is given by
\[
L=\left(
\begin{array}{c}
\pm \sqrt{\frac{2b}{\mu ^{\beta/2}}} r\\
0 \\
u \\
\frac{b \beta (\mu -1)}{\mu ^{\frac{\beta + 2}{2}}} \theta + \pm
\frac{(\beta - 2)}{2}( \sqrt{\frac{2b}{\mu ^{\beta/2}}}) u
\end{array}
\right) .
\]
As a basis for the tangent space $T_{A_{0,\pi
}}^{+,-}$, we take the vectors
\begin{equation}
\xi _{1}=\left(
\begin{array}{c}
1 \\
0 \\
0 \\
0
\end{array}
\right) ,\qquad \xi _{2}=\left(
\begin{array}{c}
0 \\
0 \\
1 \\
0
\end{array}
\right) ,\qquad \xi _{3}=\left(
\begin{array}{c}
0 \\
0 \\
0 \\
1
\end{array}
\right) .  \label{basis}
\end{equation}
The representation of the linear part $L$ relative to this basis is given
by the matrix
\[
J^{*}=\left(
\begin{array}{ccc}
\pm \sqrt{\frac{2b}{\mu ^{\beta/2}}} & 0 & 0 \\
0 & 0 & 1 \\
0 & \frac{b \beta (\mu -1)}{\mu ^{\frac{\beta + 2}{2}}}  & \pm
\frac{(\beta - 2)}{2}( \sqrt{\frac{2b}{\mu ^{\beta/2}}})
\end{array}
\right).
\]
The characteristic polynomial shows that all eigenvalues are real
and that the equilibrium is a saddle in each case.

At the equilibria $A_{\pi /2}^{\pm }$ and $A_{3\pi /2}^{\pm }$,
the same linearized system has the matrix
\[
\left[
\begin{array}{cccc}
\pm \sqrt{2b} & 0 & 0 & 0 \\
0 & (\beta - 2) \pm (\sqrt{2b}) & 0 & 0 \\
0 & 0 & 0 & 1 \\
0 & 0 & - b \beta (\mu -1) & \pm \frac{(\beta - 2)}{2}( \sqrt{2b})
\end{array}
\right].
\]
Using the vectors given in (\ref{basis}) as a basis for the tangent space
$T_{A_{\pi/2,3\pi/2 }}^{+,-}$, the linear part is given by the matrix
\[
J^{*}=\left(
\begin{array}{ccc}
\pm \sqrt{2b} & 0 & 0 \\
0 & 0 & 1 \\
0 & b \beta (\mu -1)  & \pm \frac{(\beta - 2)}{2}( \sqrt{2b})
\end{array}
\right).
\]
The eigenvalues at $A_{\pi /2}^{+}$ and at $A_{3\pi
/2}^{+}$ are
$$\sqrt{2b}, \quad
\frac{(\beta - 2)}{2}( \sqrt{2b}) + \sqrt{\frac{b}{2}[(\beta -2)^2
- 8\beta(\mu - 1)]},$$ and
$$ \frac{(\beta - 2)}{2}( \sqrt{2b}) -
\sqrt{\frac{b}{2}[(\beta -2)^2 - 8\beta(\mu - 1)]}.$$
This means that all of them are positive or have positive real part.
At $A_{\pi /2}^{-}$ and $A_{3\pi /2}^{-}$ the eigenvalues
are all negative or have negative real part.
\end{proof}

\begin{cor} The set of initial conditions leading to collisions or
ejections has positive measure.
\end{cor}
\begin{proof}
The equilibrium points $A_{\pi /2}^{-}$ ($A_{\pi /2}^{+}$) and
$A_{3\pi /2}^{-}$ ($A_{3\pi /2}^{+}$) are sinks (sources) for the
global flow, therefore their basin of attraction (repulsion) is a
3-dimensional set of collision (ejection) orbits.
\end{proof}

\begin{remark} If  $0<\mu < \frac{(\beta + 2)^2}{8\beta}$, all the
eigenvalues are real and positive. If $\mu >\frac{(\beta
+ 2)^2}{8\beta}$, there are two complex eigenvalues with
positive real part, so some orbits have the spiraling property,
i.e.\ engage into an infinite spin. For example, to have spiraling
orbits in the case $\beta=3$, it is necessary that $\mu > \frac{25}{24}$. Therefore when $\mu$ is sufficiently close to 1, no spiraling orbits exist.
\end{remark}

\begin{remark}
The inequalities in Remark 3.1 do not depend on the parameter $b$.
\end{remark}
\begin{remark}
For $\mu$ close enough to 1, though the set of collision
orbits has positive measure, there are no spiralling orbits.
\end{remark}

\section*{\large\bf IV. Saddle-Saddle Connections on $C$.}

Using ideas similar to those found in Ref. 7, we will further
study the existence of saddle-saddle connections on $C$ for $\mu>1$ and
$\beta$ of the form
\[
\beta=2+\frac{2}{1+2k} \quad \mbox{and}\quad \beta=2+\frac{1}{1+k},
\]
$k$ integer, $k\neq -1$. The reason for choosing these
values of $\beta$ will be clarified below. In the
proof we restrict ourselves to the cases $\beta=3,4$. This is
because the method requires the computation of an integral for each
value of $k$, and each integral has to be computed separately. But
the principle is the same for every such $k$.

To show that there are no heteroclinic connections is sufficient to
prove that  the stable and unstable manifolds of corresponding fixed points miss each
other. We will show that this holds true for most values of $\mu>1$.

It is now convenient to introduce different coordinates on $C$.
Since $C$ is homeomorphic to a torus, we can describe the flow
using angle variables. With the transformations
\beq
\begin{split}
 u=&{\sqrt{2b}\over \Delta^{\beta/4}}\sin\psi, \\
 v=&{\sqrt{2b}\over\Delta^{\beta/4}}\cos\psi,
\end{split}
\eeq
we can rewrite the flow on the equations of motion (\ref{eqcol}) on $C$ as
\beq
\begin{split}
\theta'=&{\sqrt{2b}\over\Delta^{\beta/4}}\sin\psi\\
\psi'=&{\beta-2\over 2}{\sqrt{2b}\over \Delta^{\beta/4}}\sin\psi+
{\beta(\mu-1)\over 4}{\sqrt{2b}\over \Delta^{\beta+4\over 4}}\sin 2\theta
\cos \psi.
\end{split}
\label{angleontorus}
\eeq
The equilibrium points in the new variables $(\theta,\psi)$ are
$A^-_{\pm \pi/2}=(\pm \pi/2,0)$, $A^+_{\pm \pi/2}=(\pm \pi/2,\pi)$,
$A_0^-= (0,0)$, $A_0^+=(0,\pi)$, $A^-_\pi=(\pi,0)$ and $A^+_\pi=(\pi,\pi)$.

Notice that if $\mu=1$ (the isotropic case), the collision
manifold is a torus for which the upper and lower circles $C^+$
and $C^-$ consist of fixed points. The equations (\ref{angleontorus})
take the form
\beq
\begin{split}
\theta'=&\sqrt{2b}\sin\psi\\
\psi'=&{\beta-2\over 2}\sqrt{2b}\sin\psi.
\end{split}
\label{angleontorus1}
\eeq
It ie easy to see that in this case there are heteroclinic orbits that
connect the critical points $A_0^+$ to $A_\pi^-$ and $A_\pi^+$ to $A_0^-$
when
\[
\beta=2+\frac{2}{1+2k},
\]
$k$ integer, $k\neq -1$, and connect the critical points
$A_{-\pi}^+$ to $A_\pi^-$ and $A_\pi^+$ to $A_{-\pi}^-$ when
\[
\quad \beta=2+\frac{1}{1+k},
\]
$k$ integer, $k\neq -1$.

Figure 2(a) depicts the collision manifold for $\mu=1$ and $\beta=3$
and the heteroclinic connections on it, while Figure 2(b) does the
same for $\mu=1$ and $\beta=4$.

\begin{center}
\begin{figure}[t!]
\begin{center}
\begin{minipage}[h]{14cm}
\parbox[t]{6.5cm}{\begin{center}
\rotatebox{0}{\resizebox{!}{5cm}{\includegraphics{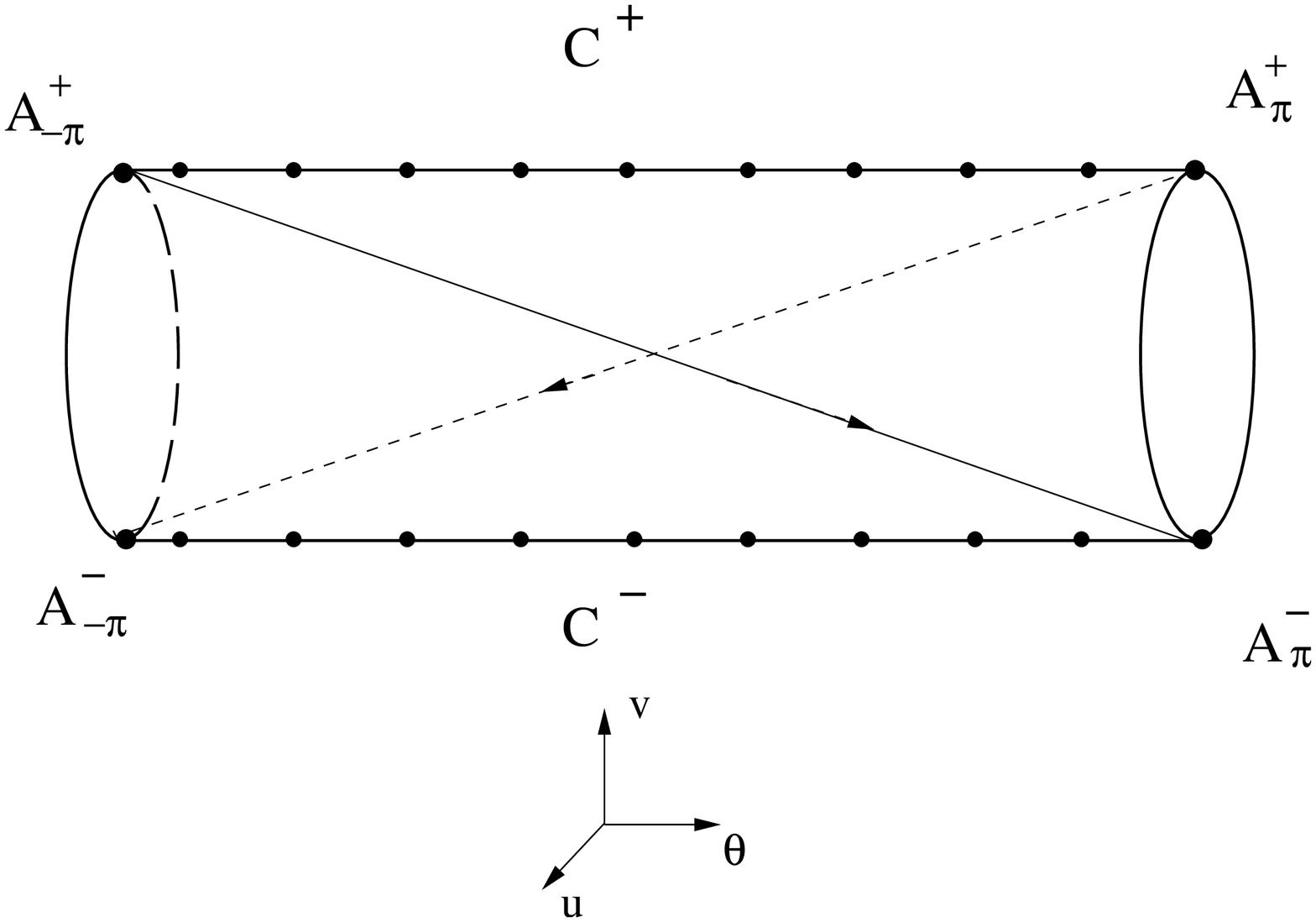}}}\end{center}
\begin{center}{\footnotesize(a)} \end{center} }
\parbox[t]{6.5cm}{\begin{center}
\rotatebox{0}{\resizebox{!}{5cm}{\includegraphics{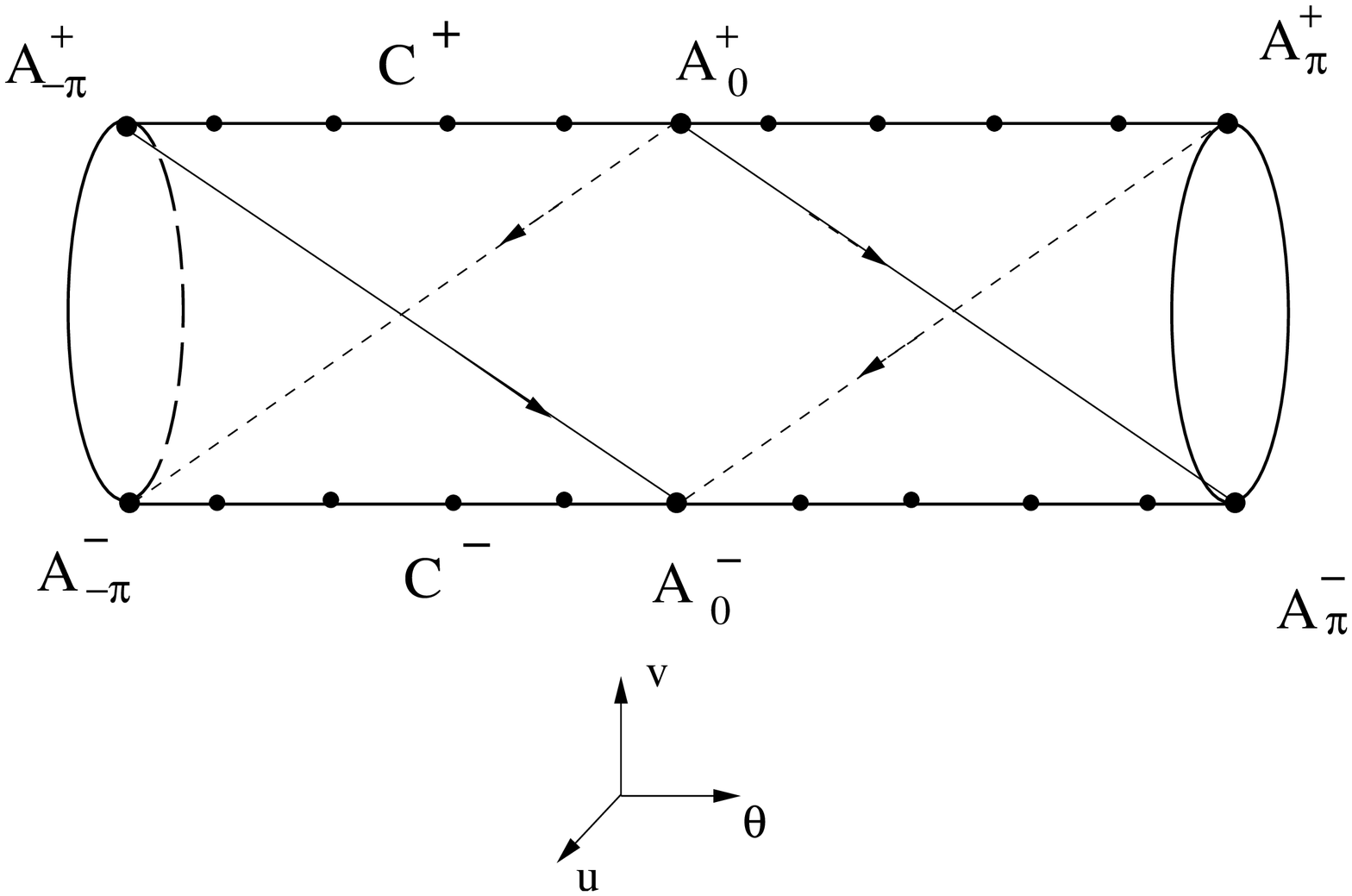}}
 } \end{center} \begin{center}{\footnotesize(b)} \end{center} }
\end{minipage}
\end{center}
\begin{center}
\caption{The collision manifold and the heteroclinic connections
for (a) $\mu=1$, $\beta=3$ and for (b) $\mu=1$, $\beta=4$.}
\end{center}
\label{cylinder}
\end{figure}
\end{center}

In the following we will show that if $\mu-1=\epsilon>0$ and small,
such saddle-saddle connections are broken and the same result holds
for an open dense set of $\mu-1=\ep>0$.

\begin{theo}
For $\beta=3$ and for an open and dense set of real numbers $\mu>1$, the
unstable manifolds at   $A^-_{-\pi}$  and $A^-_{\pi}$ miss the stable manifolds at  $A^+_\pi$ and $A^+_{-\pi}$.
For $\beta=4$ and for an open and dense set of real numbers $\mu>1$, the
unstable manifolds at $A^+_0$ and $A^+_{-\pi}$ miss the stable manifolds at $A^-_\pi$ and  $A^-_0$.

\label{Prmanifolds}
\end{theo}
\begin{proof} Eliminating the time from equations(\ref{angleontorus}) and using the fact that
\beq
{d\over dt}\left({\sqrt{2b}\over\Delta^{\beta/4}}\right)=
{\beta(\mu-1)\sqrt{2b}\over 4\Delta^{\beta+4\over 4}}\sin 2\theta,
\eeq
we have
\beq
{d\psi \over d\theta}={d \over d\tau}\left({\sqrt{2b} \over \Delta^{\beta/4}} \right){\Delta^{\beta/2}\over 2b}{\cos\psi \over \sin^2\psi}+\left(\frac{\beta -2}{ 2}\right)=F_\beta(\theta,\psi,\epsilon),
\label{Eqpsitheta}
\eeq
where $\epsilon=\mu-1$ and $\Delta=1+\epsilon\cos^2\theta$.

First consider $\beta=3$, and the unstable manifolds $W^u_3(-\pi,0)=W^u_3(\pi,0)$. When $\ep=0$, $W_3^s(\pi,\pi)$ matches $W^u_3(-\pi,0)$. Consider the branch of $W^u_3(-\pi,0)$ which contains $(0,\pi/2)$. This curve lies along the line
\beq
-2\psi+\theta=-\pi.
\eeq
When $\ep$ varies, this branch of the unstable manifold $W^u_3(-\pi,0)$ varies smoothly on $C$. Let $\zeta^3(\theta,\ep)$ denote the $\psi-$coordinate of this curve, with $\zeta^3(-\pi,\ep)=0$.

Now let $\beta=4$ and $W^u_4(-\pi,0)=W^u_4(\pi,0)$. When $\ep=0$, $W_4^s(-\pi,0)$ matches $W^u_4(0,\pi)$. Consider the branch of $W^u_4(-\pi,0)$ that contains $(-\pi/2,\pi/2)$. This curve lies along
the line
\beq
-2\psi+2\theta=-2\pi.
\eeq
As $\ep$ varies, this branch of $W^u_4(-\pi,0)$
varies smoothly on $C$. Let $\zeta^4(\theta,\ep)$ denote the $\psi-$coordinate of this curve, with $\zeta^4(-\pi,\ep)=0$. Now we need the following result, which we will prove at the end of this demonstration.

\begin{lemma}
With the above notations,
${\partial\over \partial \ep}\zeta^3(0,\pi/2)=3/4\pi>0$ and
${\partial\over \partial \ep}\zeta^4(-\pi/2,\pi/2)=\pi/2>0$.
\label{Lepartial}
\end{lemma}
From this lemma follows that $\zeta^3(0,\ep)>0$ and $\zeta^4(0,\ep)>0$
for $\ep>0$ small. Thus, it is easy to show that
$v_{1+\ep}(0,\zeta^l(0,\ep))>0$, where $l=3,4$.
Equations (\ref{Eqpsitheta}) are reversed by the transformation
\beq
(\theta,\psi)\rightarrow (-\theta,\pi-\psi).
\eeq
If $\beta=3$, the unstable manifold through $(-\pi,0)$ is mapped onto the stable manifold through $(\pi,\pi)$. Hence the stable manifold intersects the line $\theta=0$ at some point $(0,\psi_0)$ such that $v_{1+\ep}(0,\psi_0))<0$. Consequently the stable manifold misses the unstable one for $\epsilon>0$ small.

Moreover the stable and unstable manifolds intersect only for a discrete set of $\epsilon$, since they vary analytically with $\epsilon$.

Furthermore equations (\ref{Eqpsitheta}) are reversed by the transformation
\beq
(\theta,\psi)\rightarrow (-\theta-\pi,\pi-\psi).
\eeq
If $\beta=4$, the unstable manifold through $(-\pi,0)$ is mapped onto the stable manifold  through $(0,\pi)$.  Hence the stable manifold intersects the line $\theta=0$ at some point $(0,\psi_0)$ such that $v_{1+\ep}(0,\psi_0))<0$ and the stable manifold misses the unstable one for $\epsilon>0$ small.

Moreover the stable and unstable manifolds intersect only for a discrete set of $\epsilon$, since they vary analytically with $\epsilon$.

Similar arguments can be applied to the remaining stable and
unstable manifolds. This concludes the proof of Proposition 4.1.
\end{proof}

\begin{proof}[Proof of Lemma \ref{Lepartial}.]
 Observe that $\zeta^\beta$ satisfies the equation
\beq
\zeta^\beta(\theta,\ep)=\int^\theta_{-\pi} F_\beta(\eta,\zeta^\beta(\eta),\epsilon) d\eta,
\label{EqintF}
\eeq
where $F$ is given by (\ref{Eqpsitheta}). For $\ep$ small we can write
\beq
\zeta^\beta=\zeta^\beta_0(\theta)+\ep\zeta^\beta_1(\theta)+O(\ep^2).
\eeq
We also have
\beq
\begin{split}
\zeta^3_0(\theta)=&(1/2)\theta+\pi/2,\\
\zeta^4_0(\theta)=&\theta+\pi.
\end{split}
\eeq
To compute $\zeta_1(\theta)$, we can use the Taylor expansion of
(\ref{EqintF}) with respect to $\ep$ and find
\beq
\zeta^\beta_1(\theta)=\int_{-\pi}^\theta \left ( {\partial  \over \partial \ep}F_\beta(\eta,\zeta^\theta_0(\eta),0)+
 {\partial  \over \partial \psi}F_\beta(\eta,\zeta^\beta_0(\eta),0)\zeta^\beta_1(\eta)   \right )d\eta.
\eeq
Standard computations show that
\beq
{\partial  \over \partial \ep}F_\beta(\eta,\zeta^\beta_0(\eta),\ep)= \frac{\beta}{2}\,{\frac {\cos \left( \eta \right) \sin \left( \eta \right)
\cos \left( \zeta^\beta_0(\eta) \right) }{\sin \left( \zeta^\beta_0(\eta) \right) }}+O \left( \ep
 \right)
\eeq
and that
\beq
{\partial  \over \partial \psi}F_\beta(\eta,\zeta^\beta_0(\eta),\ep)= O(\ep).
\eeq
We can now compute $\zeta^\beta_1(\theta)$ as follows,
\beq
\begin{split}
\zeta^\beta_1(\theta)&
=\int_{-\pi}^\theta \left( {\partial  \over \partial \ep}F_\beta+
 {\partial  \over \partial \psi}F_\beta \zeta_1  \right )~d\eta \\&
=\frac{\beta}{2} \int_{-\pi}^\theta \,{\frac {\cos \left( \eta \right) \sin \left( \eta \right)
\cos \left( \zeta^\beta_0(\eta) \right) }{\sin \left( \zeta^\beta_0(\eta) \right) }}~d\eta.\\
\end{split}
\eeq
When $\beta=3$,
\beq
\begin{split}
\zeta^3_1(\theta)=&\frac{3}{2} \int_{-\pi}^\theta \,{\frac {\cos \left( \eta \right) \sin \left( \eta \right)
\cos \left( \eta/2+\pi/2\right) }{\sin \left( \eta/2+\pi/2\right) }}~d\eta\\
=&-\frac 9 2\,\cos \left( \frac 1 2\,\theta \right) \sin \left( \frac 1 2\,\theta \right)
+\frac 3 4\,\theta\\
&+3\, \left( \cos \left( \frac 1 2\,\theta \right)  \right) ^{3}
\sin \left( \frac 1 2\,\theta \right) +\frac 3 4\,\pi
\end{split}
\eeq
and, in particular, for $\theta=0$, we have
\beq
\zeta^3_1(0)=\frac3 4\pi=\frac{\partial}{\partial \ep}\zeta^3(0,\frac \pi 2).
\eeq
When $\beta=4$,
\beq
\begin{split}
\zeta^4_1(\theta)=&\frac{3}{2} \int_{-\pi}^\theta \,{\frac {\cos \left( \eta \right) \sin \left( \eta \right)
\cos \left( \eta+\pi\right) }{\sin \left( \eta+\pi\right) }}~d\eta\\
=&\cos \left( \theta \right) \sin \left( \theta \right) +\theta+\pi
\end{split}
\eeq and, in particular, for $\theta=-\pi/2$, we have \beq
\zeta^4_1(0)=\frac \pi 2=\frac{\partial}{\partial
\ep}\zeta^4(-\frac \pi 2,\frac \pi 2). \eeq This concludes the
proof of Lemma 4.1.
\end{proof}


\section*{\large\bf V. Escape and Capture Solutions for $h=0$}

We will further study escape (capture) solutions, i.e.\ the ones
for which $r\rightarrow \infty$ when $t\rightarrow\infty$
($t\rightarrow-\infty$). From the energy relation (\ref{eqenerg}),
we can see that for $h<0$, the radial coordinate $r$ is bounded by
the zero velocity curve ($u=0,v=0$), so escapes exist only for
$h \geq 0$. We restrict our analysis to the case $h=0$, in which
the energy relation (\ref{eqenerg}) takes the form
\begin{equation}\label{energ-h0}
u^2 + v^2 = 2r^{\beta -1} + \frac{2b}{\Delta^{\beta/2}}.
\end{equation}

With the transformations
$$\rho=r^{-1};  \quad \bar{v}=\rho^{\frac{\beta -1}{2}}v;
\quad \bar{u}=\rho^{\frac{\beta -1}{2}}u,$$
the energy relation becomes
\begin{equation}\label{ener-rho}
\bar{u}^2 + \bar{v}^2 = 2 + \frac{2b}{\Delta^{\beta/2}}\rho^{\beta
-1}.
\end{equation}

We define the {\it infinity manifold} $I_0$ as
\begin{equation}\label{infinity}
I_0 = \{ (\rho, \bar{v}, \theta, \bar{u}) | \quad  \rho =0, \quad
\bar{u}^2 + \bar{v}^2 = 2 \}.
\end{equation}
Since the variable $\theta \in S^1$, the infinity manifold $I_0$ is a torus.

\begin{remark} The infinity manifold $I_0$ is independent of the
parameter $\beta$.
\end{remark}

After rescaling the time $\tau$ with the transformation $d\tau=\rho ^{\frac{\beta
-1}{2}}ds,$ the equations (\ref{eqmot}) take the form
\begin{eqnarray}\label{eq-inf}
\frac{d\rho }{ds} &=&-\rho \bar{v},  \nonumber \\
\frac{d\bar{v}}{ds} &=&-\frac{1}{2}\bar{v}^{2}-\frac{b(\beta -
2)}{\Delta ^{\beta /2}} \rho^{\beta -1} + 1,  \nonumber \\
\frac{d\theta }{ds} &=&\bar{u}, \\
\frac{d\bar{u}}{ds} &=&-\frac{1}{2}\bar{u}\bar{v} +
\frac{b\beta(\mu - 1)\sin{2\theta}}{2 \Delta ^{\frac{\beta
+2}{2}}}\rho^{\beta -1}.\nonumber
\end{eqnarray}

Equations (\ref{ener-rho}) and (\ref{eq-inf}) are well defined on
the boundary $\rho =0$. Consequently, the phase space of the
coordinates $(\rho ,\bar{v},\theta ,\bar{u})$ can be analytically
extended to contain the manifold $I_0$. Since $d\rho /ds=0$ for
$\rho =0$, this manifold is invariant under the flow.

\begin{prop} All the equilibrium solutions of the flow given by
(\ref{eq-inf}) lie on the infinity manifold $I_0$ and they form two
circles of equilibria given by
$$\rho=0, \quad \bar{v}=\pm \sqrt{2},
\quad \theta \in S^1 \quad \bar{u}=0.$$\end{prop}

\begin{proof}
It is obvious that any point of the above circles is an equilibrium
orbit. If $\bar{v}=0$ in (\ref{eq-inf}), by
the third equation we have that $\bar{u}=0$, but this is a
contradiction with the energy relation given by (\ref{ener-rho}).
This proves the result.
\end{proof}
\bigskip

On the infinity manifold $I_0$, the equations of motion take the
form
\begin{eqnarray}\label{eq-I}
\frac{d\bar{v}}{ds} &=&-\frac{1}{2}\bar{v}^{2}+ 1 = \bar{u}^2/2,  \nonumber \\
\frac{d\theta }{ds} &=&\bar{u}, \\
\frac{d\bar{u}}{ds} &=&-\frac{1}{2}\bar{u}\bar{v}. \nonumber
\end{eqnarray}

We can now prove the following properties.

\begin{prop} The flow on $I_0$ is {\it gradient-like} with respect to the
$\bar{v}$-coordinate.\end{prop}
\begin{proof}
From the first equation in (\ref{eq-I}), we obtain that $\bar{v}^{\prime}>0$
except at equilibria, which proves the gradient-like property.
\end{proof}

In agreement with the sign of $\bar{v}$, and by similarity with
the collision manifold for $\mu=1$ studied in Section 4, we also
denote the equilibria on $I_0$ as $C^{\pm}$, respectively.

\begin{theo} \label{theorem5.1} On the infinity manifold $I_0$,
the two circles of equilibria $C^{+}$ and $C^{-}$ are normally
hyperbolic. $C^{+}$ corresponds to a sink, whereas $C^{-}$
corresponds to a source. The escape orbits are the ones having
$C^+$ as an $\omega$-limit, whereas capture orbits are the ones
having $C^-$ as an $\alpha$-limit.

\end{theo}

\begin{proof}
The proof is similar to the one of Theorem 3.1, so we will only sketch
the main steps. From equation (\ref{ener-rho}), we define
$$G(\rho,\bar{v},\theta ,\bar{u}) = \bar{u}^2 + \bar{v}^2 - 2 -
\frac{2b}{\Delta^{\beta/2}}\rho^{\beta -1}= 0.$$ Then $I_{\infty}$
is the 3-dimensional manifold given by $G^{-1}(0)$. To study the
tangent spaces to this manifold at the equilibria, we use as a
basis the same vectors (\ref{basis}). Then the linear
representation of the vector field (\ref{eq-inf}) at any equilibrium
$C^{+}$ and $C^{-}$ is given by the matrix
\[
\left(
\begin{array}{ccc}
-v_0 & 0 & 0 \\
0 & 0 & 1 \\
0 & 0 & \frac{-v_0}{2}
\end{array}
\right).
\]

Notice that for $C^+$, $v_0 = \sqrt{2}$ and for $C^-$, $v_0 =-\sqrt{2}$.
In the former case two eigenvalues are negative and one is zero, whereas
in the latter case two eigenvalues are positive and one is zero. This
completes the proof of the normal hyperbolicty and shows the existence
of infinitely many escape orbits and capture orbits.

\end{proof}

The next result, which is a direct consequence of Theorem
\ref{theorem5.1}, characterizes the flow on the infinity manifold.

\begin{cor}
For $h=0$, the infinity manifold $I_0$ is foliated by heteroclinic
orbits between $C^-$ to $C^+$.
\end{cor}


The flow on $I_0$ given by equations (\ref{eq-I}) is easy to draw.
Because on $I_0,$ $\bar{u}^2 + \bar{v}^2 = 2$, we can introduce
the angular variable $\psi$ with the transformation
$$\bar{u}= \sqrt{2}\cos{\psi}, \quad \bar{v}=\sqrt{2}\sin{\psi}.$$
On $I_0$, the equations of motion take the form
$$\dot{\psi}= -2\sqrt{2}\cos{\psi}, \quad \dot{\theta}= \sqrt{2}\cos{\psi}.$$ From here we
obtain that
$$\frac{d\psi}{d\theta} = -2.$$
On $I_0$ we can also study the projection of the flow on the $v-\theta$
plane, which is given by
$$\frac{dv}{d\theta} = \frac{\sqrt{2-v^2}}{2},$$
whose solution is $v(\theta)=\sqrt{2}\sin{\frac{\theta + k}{2}}$,
where $k$ is a constant determined by the initial condition.


\section*{\large\bf VI. Integrability for $\beta=2$}

We will further study the problem for $\beta=2$ and
deal with a Hamiltonian (\ref{hamf}) of the form
 \beq
H_2=\frac{1}{2}{\bf p}^2-\frac{1}{\sqrt{x^2+y^2}}-\frac{b}{\mu x^2+y^2}.
\eeq
With the notation $\epsilon=\mu-1$, the Hamiltonian expressed in polar coordinates becomes
\beq
 H_2=\frac{p_r^2}{2}+\frac{p_\theta^2}{2r^2}-{1\over r}-{b\over r^2(1+\epsilon\cos^2(\theta))}.
\eeq

The corresponding system is integrable since it admits another first integral independent of the Hamiltonian, namely
\beq
G=\frac{p_\theta^2}{2}-\frac{b}{1+\epsilon\cos^2(\theta)}.
\eeq
Indeed,
\beq
\{H_2,G\}=\frac{\pd H_2}{\pd \theta}\frac{\pd G}{\pd p_\theta}-
\frac{\pd H_2}{\pd p_\theta}\frac{\pd G}{\pd \theta},
\eeq
and since
\beq
\frac{\pd  H_2}{\pd \theta}=-\frac{\ep\sin(2\theta)}{r^2(1+\ep\cos^2(\theta))^2}, \quad \frac{\pd  H_2}{\pd p_\theta}=\frac{p_\theta}{r^2}
\eeq
and
\beq
\frac{\pd G}{\pd \theta}=-\frac{\ep\sin(2\theta)}{(1+\ep\cos^2(\theta))^2}, \quad \frac{\pd G}{\pd p_\theta}={p_\theta}
\eeq
the Poisson's bracket is $\{H_2,G\}=0$, so $G$ and $H_2$ are linearly independent.

The existence of the integral $G$ is not surprising. Indeed it is well known that (see Ref. 19) given an Hamiltonian
\beq
\bar H=\frac{p_r^2}{2}+\frac{p_\theta^2}{2r^2}+U(r,\theta),
\eeq
the corresponding Hamilton-Jacobi equation
\beq
\frac{ \pd S}{\pd t}+\bar H\left(r,\theta;\frac{\pd S}{\pd r},\frac{\pd S}{\pd \theta}\right)=0
\eeq
(where $S=S(r,\theta,t)$ is the action expressed as function of the coordinates and time, $\pd S/\pd r=p_r$ and $\pd S/\pd \theta=p_\theta$) can be solved by separation of variables if
the potential energy is of the form
\beq
U=a(r)+\frac{b(\theta)}{r^2}.
\eeq
Since the Hamiltonian is time independent we take $S(r,\theta,t)=S_0(r,\theta)-Et$ (where $E$ is a constant), and the Hamilton-Jacobi equation for $S_0$ becomes
\beq
\frac 1 2\left(\frac{\pd S}{\pd r}\right)^2+a(r)+\frac{1}{2r^2}\left[
\left(\frac{\pd S}{\pd \theta}\right)^2+2b(\theta)\right]=E.
\eeq
Looking for a solution of the form
\beq
S_0=S_1(r)+S_2(\theta),
\eeq
we find for $S_1$ and $S_2$ the equations
\beq
\begin{split}
&\left(\frac{d S_2}{d\theta}\right)^2+2b(\theta)=2G,  \\
&\frac1 2\left(\frac{dS_1}{dr} \right)^2+a(r)+\frac{2G}{2r^2}=E,
\end{split}
\eeq
which define two independent integrals. Solving these equations leads to a solution of the Hamilton-Jacobi equations and thus to a general solution
of the equations of motion. The above technique applies to the Hamiltonian ${H_2}$ and to the additional integral $G$.

This approach also shows that the Hamiltonian system given by
$H_2$ is integrable by quadratures. The same conclusion can be reached
by directly applying the Liouville-Arnold theorem.\r{2}

\section*{\large\bf VII. The Collision Manifold for $\beta=2$}

For $\beta=2$ the equation of motion (\ref{eqmot}) in McGehee coordinates
can be written as
\beq 
\begin{split}
r'=& uv\\
v'=&2r^2h+r\\
\theta'=&u\\
u'=& \ep b\sin(2\theta) \Delta^{-2},
\end{split}
 \label{regeqmotion} \eeq where the prime denotes
differentiation with respect to $\tau$ and
$\Delta=1+\ep\cos^2(\theta)$. In McGehee coordinates, the energy
relation takes the form \beq u^2+v^2-2r-2b\Delta^{-1}=2r^2h, \eeq
where $h$ is the energy constant. The first integral $G$ can be
written as
\beq g=\frac 1 2(u^2-2b\Delta^{-1}), \eeq
where $g$ is also constant along orbits.

The vector field (\ref{regeqmotion}) is analytic on the boundary $r=0$,
since $r$ no longer occurs in the denominators of the vector field.
The collision manifold reduces to
\beq
C=\{(r,\theta,v,u):r=0,u^2+v^2=2b\Delta^{-1}\}.
\eeq
This shows that $C$ is homeomorphic to a torus. The  restriction of
equations (\ref{regeqmotion}) to $C$ yields the system
\beq
\begin{split}
v'=&0\\
\theta'=&u\\
u'=&\ep b \Delta^{-2}\sin 2\theta.
\end{split}
\eeq

All nonequilibrium orbits on $C$ are periodic. Comparing the
collision manifold and vector field above with the corresponding
ones in Ref. 6, we see that the collision manifold and the
flow for $\beta=2$ are identical to the ones of the anisotropic
Manev problem.

\section*{\large\bf VIII. Heteroclinic orbits for $\beta=2$ and $h=0$}

The main goal of this section is to study the infinity manifold
for $h=0$ and the heteroclinic orbits that connect the collision
and infinity manifold. First notice that for $h<0$ the motion is
bounded and therefore there is no infinity manifold. More precisely,
we have the following result.
\begin{prop}
If $h<0$ the motion is bounded by the zero velocity curve
\beq
\left(r_0=\frac{-1+\sqrt{1-4hb\Delta^{-1}}}{2h},v=0,\theta,u=0 \right).
\label{zerovelcurve}
\eeq
\end{prop}
\begin{proof}
Obviously  $u=0$ and $v=0$ if and only if $u^2+v^2=0$. Also $v=0$ implies $r'=0$. Using the energy relation we can draw the conclusion that $u^2+v^2=0$ if and only if $2r^2h +2r+2b\Delta^{-1}=0$. This quadratic equation has the solutions
\beq
 r=\frac{-1\pm\sqrt{1-4hb\Delta^{-1}}}{2h}.
\eeq
Since $r \geq 0$ and $h<0$ the only valid solution is the one with the minus sign. This shows that (\ref{zerovelcurve}) is the zero velocity curve. The fact that the motion is bounded by this curve follows from the fact that if $r>r_0$ then $u^2+v^2<0$.
\end{proof}

To describe the behavior of the solution at infinity we need to study
the equations (\ref{eq-inf}) with $\beta=2$ and $h=0$, that is
\beq
\begin{split}
\dot \rho=&-\rho \bar v\\
\dot {\bar v}=&-\frac{1}{2}\bar v^2+1\\
\dot \theta=&\bar u\\
\dot {\bar u}=&-\frac 1 2\bar v~ \bar u+\ep b\rho\sin2\theta\Delta^{-2},
\end{split}
\label{eqinfinity}
\eeq
where the dot denotes differentiation with respect to the new time variable
$s$. The energy relation is
\beq
\bar u^2+\bar v^2-2-2b\rho\Delta^{-1}=0
\label{energyrel}
\eeq
and the other first integral can be written as
\beq
\bar u^2-2b\rho\Delta^{-1}=2\rho g.
\label{integral}
\eeq

The infinity manifold is the set
\beq
I_0=\{(\rho, \bar v,\theta,\bar u)|\rho=0,~ \bar u^2+\bar v^2=2, ~\bar u^2=0\},
\eeq
i.e. the points in phase space that satisfy the condition $\rho=0$, the energy relation and the additional conservation relation. $I_0$ is a disconnected manifold formed by the union of two disjoint circles of fixed points lying in parallel planes.

The first two equations of the system (\ref{eqinfinity}) are independent
from the others, and we would like to determine $\rho$ and $\bar v$. If
we set $\bar v=\pm \sqrt{2}$, then $\rho=\exp(\mp\sqrt{2}(s-s_0))$ is a solution of the two equations mentioned above.
If $\bar v=\pm \sqrt 2$ the energy integral (\ref{energyrel}) gives the condition
\beq
\bar u=\pm \sqrt{2b\rho\Delta^{-1}}.
\label{condition}
\eeq
Moreover the previous condition and  equation (\ref{integral}) impose $g=0$.
Differentiating (\ref{condition}) with respect to $s$ we obtain
\beq
\dot{\bar u}=\pm\sqrt{2b}\left(\frac{\Delta^{-1/2}}{2\rho^{1/2}}\dot\rho+
\frac \epsilon 2\bar u\sqrt{\rho\Delta}~ \frac{\sin 2\theta}{\Delta^2}\dot\theta\right)
\eeq
and using the first equation in (\ref{eqinfinity}) and equation (\ref{condition}) it follows that
\beq
\dot {\bar u}=-\frac 1 2\bar v~ \bar u+\ep b\rho\sin2\theta\Delta^{-2}.
\eeq
This shows that system (\ref{eqinfinity}) admits solutions with $\bar v=\pm \sqrt{2}$.

The other solutions of the first two equations of the system  (\ref{eqinfinity}) can be found by dividing the second equation
by the first.  This leads to the equation
\beq
\frac{\bar v}{\rho}=\frac 1 2\frac{\bar v}{\rho}-\frac{1}{\rho \bar v},
\eeq
which can be solved by separating the variables. This leads to
\beq
\int_{\rho_0}^\rho \frac{d\xi}{\xi}=\int_{\bar v_0}^{z}\frac{z~dz}{\frac 1 2 z^2-1}
\eeq
and consequently yields
\beq
\rho=\left (\frac{\rho_0}{\bar v_0^2-2}\right)(\bar v^2-2),
\label{rhozero}
\eeq
where $\rho_0$ and $v_0$ are initial conditions and $\rho\geq 0$, since $\rho<0$ has no physical meaning. Now we can prove the following result
relative to the heteroclinic solutions connecting the infinity and the
collision manifolds.


\begin{theo}
The solutions whose $\omega$-limit set belongs to the infinity manifold
have the $\alpha$-limit set contained in the collision manifold. In particular,
the following properties take place:
\begin{enumerate}
\item If  $\bar v=\sqrt{2}$ ($\bar v=-\sqrt{2}$), the above solutions belong to the unstable (stable) manifold of one of the periodic orbits on the equator of the collision manifold.
\item If $0<\sqrt{1/k}< \sqrt{2b} \mbox{ and } \sqrt{1/k}\neq\sqrt{2b/\mu}
$ with $k=\rho_0/(\bar v_0^2-2)$, the above solutions belong to the unstable (stable) manifold of the periodic orbits on the collision manifold with $v=\sqrt{1/k}$ ($v=-\sqrt{1/k}$).

\item If $\sqrt{1/k}= \sqrt{2b/\mu}$ ($-\sqrt{1/k}= \sqrt{2b/\mu}$), then the
above solutions belong to the unstable (stable) manifold of one of the fixed points $A_{0}^+$, $A_{\pi}^+$
 ($A_{0}^-$, $A_{\pi}^-$).
\item If $\sqrt{1/k}= \sqrt{2b}$ (-$\sqrt{1/k}= \sqrt{2b}$), then the above solutions belong to the unstable (stable) manifold of one of the fixed points $A_{-\pi/2}^+$, $A_{\pi/2}^+$
($A_{-\pi/2}^-$, $A_{\pi/2}^+$).
\end{enumerate}
\end{theo}
\begin{proof}
To prove (1) observe that  $v=\bar v/\rho^{1/2}=\pm\sqrt{2}/\rho^{1/2}$.
Thus $\lim_{s\rightarrow \infty} v=\lim_{s\rightarrow \infty}\pm\sqrt{2}/\rho^{1/2}=\lim_{\tau\rightarrow \infty}\pm\sqrt{2}/\rho^{1/2}=\lim_{\tau\rightarrow \infty} v= 0$, since $\lim_{s\rightarrow\infty}\tau= \infty$. To prove (2),(3) and (4) we consider the limit $\lim_{\tau \rightarrow\infty  }v$, which with the help of equation (\ref{rhozero}) becomes
\[\lim_{\tau \rightarrow \infty}v=\lim_{\tau \rightarrow \infty}\frac{\bar v}{\rho^{1/2}}=\pm\lim_{\tau\rightarrow \infty}\sqrt{\frac{\rho+2k}{k}}\rho^{-1/2}=\pm\sqrt{\frac 1 k}.
\]
Moreover from the energy relation $\bar u^2+\bar v^2-2-2b\rho/\Delta=0$ and the fact that $\bar v^2=(\rho+2k)/k$ it easy to see that
\[
\bar u^2=-\rho\left( \frac 1 k-\frac{2b}{\Delta}\right),
\]
and since $\rho>0$ and $\bar u^2\geq 0$, we have
\[
\frac 1 k\leq \frac{2b}{\Delta}\leq 2b.
\]
Consequently we have shown that
\beq
0<\sqrt{\frac 1 k}\leq \sqrt{2b}.
\eeq
In particular it is clear that when $\sqrt{1/k}=\sqrt{2b/\mu}$ ($-\sqrt{1/k}=\sqrt{2b/\mu}$) the solutions lie on the unstable (stable) manifold of one of the fixed points $A_{0}^+$, $A_{\pi}^+$ ($A_{0}^-$, $A_{\pi}^-$).
Similarly when  $\sqrt{1/k}= \sqrt{2b}$ (-$\sqrt{1/k}= \sqrt{2b}$), the solutions lie on the unstable (stable) manifold of one of the fixed points $A_{-\pi/2}^+$, $A_{\pi/2}^+$ ($A_{-\pi/2}^-$, $A_{\pi/2}^-$).
In the remaining cases the solutions  lie on the unstable (stable) manifold of the periodic orbits on the collision manifold with $v=\sqrt{1/k}$ ($v=-\sqrt{1/k}$).
\end{proof}

\section*{\large\bf IX. A Perturbative Approach}

In this and the next section we regard the problem in a different way
by considering the Hamiltonian
\beq
{\mathcal H}_\beta={1\over 2}{\bf p}^2-{1\over r}-{b\over r^\beta} +\ep b {\beta\cos^2\th \over 2r^\beta} \equiv{\mathcal H}^0+b {\mathcal W}^1_\beta(r,\th)+\ep b {\mathcal W}^2_\beta(r,\th),
\label{hampertgen}
\eeq
where $\beta>3/2$, $b\ll1$, $\epsilon\ll 1$. 
This is the original Hamiltonian $H_\beta$ (defined in (\ref{hamf})) truncated to the second order. Consider, as in Ref. 5, the parabolic solutions of the unperturbed ($\ep=0$) problem (\ref{hampertgen}), defined by the Hamiltonian ${\mathcal H}^0$ of the classical Kepler problem that are on the zero-energy manifold and play the role of homoclinic solutions
corresponding to the critical point at infinity, i.e., $r=\infty,
\ \.r=0$. These solutions satisfy the equations:

\beq
\.r=\pm {\sqrt{2r-k^2}\over r},\qq\qq\.\th={k\over{r^2}},
\label{eqr}
\eeq
where $k\not=0$ is the (constant) angular momentum and the sign $-$
(resp. $+$) holds for $t<0$ (resp. $t>0$). From (\ref{eqr}) we get
\beq
\begin{split}
\pm t&={k^2+r\over 3}\sqrt{2r-k^2} + {\rm const.} \\
 \th&=\pm 2\arctan{\sqrt{2r-k^2}\over \sqrt{k^2}}+ {\rm const.}
\label{eqrtheta}
\end{split}
\eeq
We denote by \beq R=R(t) \qq {\rm and}\qq \Th=\Th(t)
\label{rtheta} \eeq the expressions giving the dependence of $r$
and of $\th$ on the time $t$. These are obtained by ``inverting''
the equations (\ref{eqrtheta}) with the conditions $R(0)=r_{min}
=k^2/2$ and $\Th(0)=\pi$. Let us emphasize that, as in
Ref. 5, it is not necessary to have the
explicit form of these functions. But it is important to remark
that $R(t)$ is even and $\Th(t)$ is odd, both in the
time variable.

The parabolic orbits can also be described in parametric form.\r{19}
If $p=k^2\neq 0$, we can write
\beq
r=\frac p 2 (1+\eta^2), \quad t= \frac{p^{3/2}} 2 \eta(1+\frac {\eta^2} 3), \quad \eta=\tan\frac \theta 2.
\label{parametric}
\eeq
We also have
\beq
\cos 2\theta= 2\frac{(1-\eta^2)^2}{(1+\eta^2)^2}-1.
\eeq
We will further use these remarks to apply the Melnikov method.

\section*{\large\bf X. The Melnikov Method}

Consider the problem defined in (\ref{hampertgen}). The homoclinic manifold, i.e., the set of solution of the unperturbed equation which are doubly asymptotic to $r=\infty, \.r=0$, is given for each value $k \neq 0$ of the angular momentum by the 2-dimensional manifold described by the family of solutions $r=R(t-t_0), \vartheta = \Th (t-t_0) + \th_0$, where $R(t)$ and $\Th(t)$ have been defined in (\ref{rtheta}), with arbitrary $t_0$, $\th_0$.

It is clear from equation (\ref{hampertgen}) that the first order in $b$ of the perturbation (i.e. the term $b {\mathcal W}_\beta^1$) does not contribute to the Melnikov integrals. Actually the perturbed Hamiltonian truncated at the first order is integrable.


The perturbation resulting from a small anisotropy vanishes as
$r \rightarrow \infty$ since ${\mathcal W}^2_\beta(r,\th) \sim 1/r^\beta$
with $\beta>3/2$ satisfies condition (18) in Ref. 5. Moreover, since the Hamiltonian ${\mathcal H}_0+\epsilon {\mathcal W}_\beta^2$ is integrable, the terms in $b^n$ for $n\geq 2$ do not contribute either.
This allows us to write the first non-vanishing effect on the Melnikov integral in the same form as in Ref. 5, with the
difference that here we can drop the dependence on time,
\beq \begin{split}
 M_1(\th_0)=&\intR \!\!
\Big[\.R(t) {\pd {\mathcal W}_\beta^2\big(R(t),\Th(t)+\th_0\big)\ov{\pd r}} \\
&+\.\Th(t) {\pd {\mathcal W}_\beta^2\big(R(t),\Th(t)+\th_0\big)\ov{\pd\th}} \Big]\d t
=0
\label{melint0}
\end{split}
\eeq
\beq
M_2(\th_0)=
\intR {\pd {\mathcal W}_\beta^2\big(R(t),\Th(t)+\th_0\big)\ov{\pd \th}}~\d t=0.
\label{melint}
\eeq
Since the perturbation ${\mathcal W}_\beta^2$ vanishes as $t \rightarrow \pm \infty$, the first Melnikov condition can be written as

\beq
M_1(\th_0)=\intR {\pd {\mathcal W}_\beta^2\big(R(t),\Th(t)+\th_0\big)\ov{\pd t}}
~\d t \equiv 0. \eeq
The above integral is identically zero because the perturbation ${\mathcal W}_\beta^2$ is not time dependent.
This simplifies our discussion since we must only find the zeroes of (\ref{melint}).

It is significant to remark, and easy to verify, that these conditions can be written also in terms of the first integrals of the unperturbed problem as

\beq M_1(\th_0) = \intR \{H_0,{\mathcal W}_\beta^2\}\big(\ldots)\ \d t = 0\eeq
and
\beq M_2(\th_0) = \intR \{K,{\mathcal W}_\beta^2\}\big(\ldots)\ \d t = 0,\eeq
where $(\ldots)$ represents the homoclinic orbit.

This resembles some properties obtained for the  Gyld\'en problem (see Refs. 5,16) and is related to the symmetries of the problem. In the Gyld\'en problem there is a perturbation that doesn't depend on the angle $\th$, but depends on time. This means that
the perturbation destroys the time homogeneity, so the Hamiltonian is not an integral of motion anymore, but doesn't destroy the rotational invariance, so the angular momentum is still conserved. Therefore only one condition is given by (\ref{melint0}). However, the anisotropy destroys the rotational symmetry  but not the  homogeneity of time, so we are left with the condition (\ref{melint}).

Here the Melnikov condition for $M_2$ becomes
\beq
M_2(\th_0) = \frac{\beta}{2}\intR \frac{\sin[ 2( \Th(t)+\th_0)]}{  R(t)^\beta}d t = 0.
\eeq
Using some trigonometry the integral can be written as
\beq M_2(\th_0)=  I_1 \cos 2 \th_0  +I_2  \sin 2\th_0, \eeq
where $I_1$ and $I_2$ are defined by
\beq \begin{split}
I_1=& \frac{\beta}{2}\intR  \frac{\sin 2 \Th(t)}{  R(t)^\beta} dt,\\
I_2=& \frac{\beta}{2}\intR  \frac{\cos 2 \Th(t)}{  R(t)^\beta} dt.
\end{split}
\eeq
Recall that $R(t)$ is an even function of time and $\Th(t)$ is an odd function. This implies  that the integrand of $I_1$ is an odd function. Therefore $I_1 \equiv 0$, and $M_2$ can be rewritten as
\beq M_2(\th_0)= I_2 \sin 2\th_0. \eeq

Using the parametric form of the parabolic orbits defined in equations  (\ref{parametric}), and observing that $dt=\frac{p^{3/2}}{2}(1+\eta^2)d\eta,$ we can write
\beq
I_2=2^{\beta-1}p^{3/2-\beta}\frac{\beta}{2}\intR \frac{1}{(1+\eta^2)^{\beta-1}}\left (\frac{2(1-\eta^2)^2}{(1+\eta^2)^2}-1\right) d\eta.
\eeq
Computing the integral, we find that
\beq\begin{split}
I_2=&\frac{2^{\beta-1}p^{3/2-\beta}\beta}{2\Gamma(\beta-1)}\sqrt{\pi}\left[ \Gamma(\beta-\frac 3 2 )\left(\frac{3}{2(\beta-1)\beta}-1\right)\right.\\ &+2\left. \frac{\Gamma(\beta+\frac 1 2 )- \Gamma(\beta -\frac 1 2 )}{(\beta-1)\beta}\right],
\end{split}\eeq
where $\Gamma(z)$ is Euler's Gamma function.
Thus $I_2(\beta)$ is an analytic function in $\beta$ for $\beta>3/2$, since $\Gamma(z)$ is analytic for $z>0$.
Recall that the Gamma function can be expressed as
\beq
\Gamma(z)=\lim_{n\to \infty} \frac{n!n^z}{z(z+1)\dots (z+n)}
\eeq
if $z\neq 0,-1,-2,-3,\dots$
Using this form of the Gamma function, and letting $A=2^{\beta-2}p^{3/2-\beta}$, we find that
\beq
\begin{split}
I_2&=A\lim_{n\to \infty}n^{\frac 3 2 }\sqrt{\pi}  \frac{\beta(\beta+1)(\beta+2)\dots(\beta-1+n)}{(\beta-\frac 3 2)(\beta- 1 2)\dots (\beta +\frac 1 2 +n)}(\beta^2-5\beta+6)\\
&=\frac{A\sqrt{\pi}~ \Gamma(\beta+\frac 1 2 )}{
(\beta-1)(\beta-\frac 3 2)(\beta- \frac 1 2)\Gamma(\beta -1)}(\beta^2-5\beta+6),
\end{split}
\eeq
which is zero if and only if $\beta^2-5\beta+6=0$, since the Gamma function $\Gamma(z)$ is always positive for $z>0$ and therefore the first factor never vanishes.
Consequently $I_2$ vanishes if and only if $\beta^2-5\beta+6=0$, namely
for $\beta=2$ or $\beta=3$, see Figure \ref{Gamma}.

\begin{figure}[t]
\begin{center}
\resizebox{!}{6cm}{\includegraphics{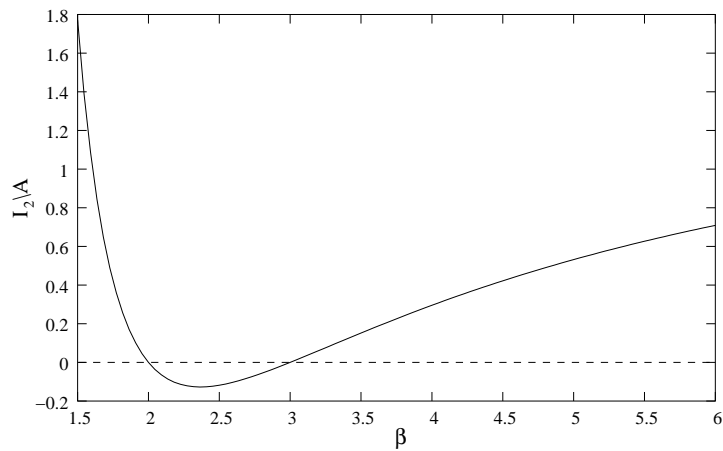}}
\end{center}
\caption{$I_2/A$ as a function of $\beta$.%
}
\label{Gamma}
\end{figure}
We have thus proved the following result.
\begin{theo}
For every $\beta>3/2$, $\beta\neq 2,3$, the system given by the Hamiltonian (\ref{hampertgen}) exhibits chaotic dynamics on the zero energy manifold.
\end{theo}
This type of chaotic behavior is induced by a chain of infinitely many intersections of the positively and negatively asymptotic sets to the
critical point at infinity. The Smale-Birkhoff theorem does not directly
apply to this situation, which is degenerate. But it is well known that
the existence of Smale horseshoes and positive topological entropy can
arise in the case of nonhyperbolic equilibria (see for example Ref. 4, where the problem is studied in the case of area-preserving diffeomorphisms).

\qq\qq

\section*{\large\bf References}
\parindent 0 pt
\footnotesize
\r{1}
 A.\ Albouy, Lectures on the Two-Body Problem, in  Classical and
Celestial Mechanics---The Recife Lectures, edited by  H.~Cabral and F.~Diacu, (Princeton Univ.\ Press, Princeton, NJ, 2002), pp.\ 63-116,.

\r{2}
V.I.\ Arnold, Mathematical Methods of Classical Mechanics, (Springer
Verlag, Berlin, 1978).

\r{3}
M.~Arribas, A.~Elipe and A.~Riaguas, ``Nonintegrability of anisotropic quasihomogeeous Hamiltonian systems",  Mechanics Research Comm.
{\bf 30}, 209-216 (2003).

\r{4}
K.~Burns and H.~Weiss, ``A geometric criterion for positive topological entropy", {\it Comm.\ Math.\ Phys.} {\bf 172}, 95-118 (1995).

\r{5}
G.~Cicogna and M.~Santoprete, ``An approach to Melnikov theory in celestial mechanics",  J. Math Phys.{\bf 41}, 805-815 (2000).

\r{6}
S.~Craig, F.~Diacu, E.A.~Lacomba and E.~P\'erez-Chavela, ``On the anisotropic Manev problem", J. Math. Phys. {\bf  40}, 1359-1375 (1999).

\r{7}
R.L.~Devaney, ``Collision orbits in the Anisotropic Kepler Problem",  Inventiones Math. {\bf 45}, 221-251 (1978).

\r{8}
F.~Diacu,  Singularities of the $N$-Body Problem---An Introduction to
Celestial Mechanics, (Les Publications CRM, Montr\'eal, 1992).

\r{9}
F.~Diacu, ``Near-collision dynamics for particle systems with quasihomogeneous potentials",  J. Diff. Eqn.
{\bf 128} (1996), 58-77.

\r{10}
F.~Diacu, ``Stability in the anisotropic Manev problem",  J.\ Phys.\ A {\bf 33}, 6573-6578 (2000).

\r{11}
F.~Diacu and P.~Holmes,  Celestial Encounters---The Origins of
Chaos and Stability, (Princeton Univ.\ Press, Princeton, NJ, 1996).

\r{12}
F.~Diacu, V.~Mioc and C.~Stoica, ``Phase-space structure and regularization of Manev-type problems",  Nonlinear Analysis, {\bf 41}, 1029-1055 (2000).

\r{13}
F.~Diacu, E.~P\'erez-Chavela and M.~Santoprete, ``Saari's conjecture for the collinear $N$-body problem",  Trans.\ Amer.\ Math.\ Soc. (to appear).

\r{14}
F.~Diacu and M.~Santoprete, ``Nonintegrability and chaos in the
anisotropic Manev problem",  Physica D {\bf 156}, 39-52 (2001).

\r{15}
F.~Diacu and M.~Santoprete, ``On the global dynamics of the anisotropic
Manev problem",  Physica D (to appear).

\r{16}
F.~Diacu and D.~Selaru, ``Chaos in the Gyld\'en problem",
 J.\ Math.\ Phys. {\bf 39}, 6537-6546 (1998).

\r{17}
M.~Gutzwiller, ``The anisotropic Kepler problem in two dimensions",
J.\ Math.\ Phys. {\bf 14}, 139-152.

\r{18}
M.~Gutzwiller, ``Bernoulli sequences and trajectories in the
anisotropic Kepler problem",  J.\ Math.\ Phys. {\bf 18}, 806-823 (1977).

\r{19}
L.D.~Landau and E.M.~Lifshitz,  Mechanics, (Pergamon, Oxford, 1969).

\r{20}
R.~McGehee, ``Triple collision in the collinear three-body
problem", {\it Inventiones Math.} {\bf 27}, 191-227 (1974).

\r{21}
E.~P\'erez-Chavela and L.V.~Vela-Ar\'evalo, ``Triple collision
in the quasihomogeneous collinear three-body problem",
{\it J.\ Diff. Eqn.} {\bf 148}, 186-211 (1998).

\r{22}
M.~Santoprete, ``Symmetric periodic solutions of the anisotropic Manev problem",
 J.\ Math.\  Phys. {\bf 43}, 3207-3219 (2002).

\end{document}